%% file: BookChapter-QuantumWasserstein-arXiv.tex
\documentclass[11pt]{article} 

\usepackage{Liz}
\usepackage{authblk}
\usepackage{xcolor}
\usepackage{tikz}

\graphicspath{{figs/}}

 \usepackage{cite}

\title{Computing Wasserstein Distance for Persistence Diagrams on a Quantum Computer}
\author[1]{Jesse J. Berwald}
\author[1]{Joel M.~Gottlieb}
\author[2,*]{Elizabeth Munch}
\affil[1]{D-Wave Systems, Inc., 3033 Beta Avenue, Burnaby, British Columbia, Canada V5G 4M9}
\affil[2]{Dept.~of Computational Mathematics, Science and Engineering; and Dept.~of Mathematics.
	Michigan State University, East Lansing, MI. }
	\affil[*]{Corresponding author, \url{muncheli@egr.msu.edu}}
\date{}

\begin{document}
\maketitle

\begin{abstract}
	Persistence diagrams are a useful tool from topological data analysis that provide a concise description of a filtered topological space.
	They are even more useful in practice because
	they come with a notion of a metric, the Wasserstein distance (closely related to but not the same as the homonymous metric from probability theory).
	Further, this metric provides a notion of stability; that is, small noise in the input causes at worst small differences in the output.
	In this paper, we show that the Wasserstein distance for persistence diagrams can be computed through quantum annealing.
	We provide a formulation of the problem as a {\em quadratic unconstrained binary optimization problem}, or QUBO, and prove correctness.
	Finally, we test our algorithm, exploring parameter choices and problem size capabilities, using a D-Wave 2000Q quantum computer.
\end{abstract}


\input{intro.tex}

\input{bkgd.tex}

\input{method.tex}

\input{equivalence.tex}
\input{experiments.tex}

\input{conclusion.tex}

\paragraph{Acknowledgements:}
 The work of EM was supported in part by NSF grants DMS-1800446 and CMMI-1800466.
 JJB gratefully acknowledges support from the Institute for Mathematics and its Applications at the University of Minnesota.
The authors thank D-Wave Systems for computing time on their machine.

 \bibliographystyle{ieeetr}
 \bibliography{QuantumTDA_arXiv}

\end{document}

%% file: intro.tex

\section{Introduction}

The field of Topological Data Analysis (TDA)
\cite{Munch2017,Oudot2017a,Carlsson2009,Ghrist2008, Otter2017, Edelsbrunner2010}
 has grown exponentially in recent years.
TDA consists of a suite of tools, derived from ideas in the mathematical field of topology, which can be used to find shape in data in a way that is quantifiable, comparable, robust, and concise.
One of the most prominent tools in the field is persistent homology \cite{Edelsbrunner2002,Zomorodian2004,Robins2000},
which encodes the changing homology of filtered topological space in a persistence diagram and can be used to understand the structure of the space.
One common assumption for the starting input data to compute persistence is a finite point cloud (really, a finite metric space) as in the example of \cref{fig:ExamplePersistence}, but many other types of input data can be used.
No matter the input data, the output persistence diagram is a collection of points in the plane above the diagonal, which can reveal the relative prominence of homological features in the data set.
Persistent homology has been successfully applied to data in many disparate domains, including but not limited to,
dynamical systems and time series analysis \cite{Colonnello2018,Khasawneh2015,Khasawneh2017,Khasawneh2018,Khasawneh2018a,Perea2015,Perea2016,Tralie2017, Emrani2014,Berwald2014a,Berwald2014b};
neuroscience \cite{Sizemore2018a,Giusti2016,Saggar2018,Giusti2015,Sizemore2018b};
plant biology \cite{Li2018, Chambers2018};
image processing \cite{Carlsson2007, Perea2014a,Asaad2017};
and genetics \cite{Chan2013,Emmett2014,Emmett2014a,Camara2016,Camara2016a}.

Persistent homology is particularly useful because it is \textit{stable}.
That is, there are metrics on persistence diagrams for which small variations in the input data result in quantifiably small variations in the persistence diagram output \cite{Cohen-Steiner2007,Cohen-Steiner2010}.
This is important for persistence on real data because all data comes with noise; stability means we can still trust the output persistence diagram with a reasonable degree of certainty.

Two closely related metrics for persistence diagrams have this stability property.
The first is the bottleneck distance  \cite{Cohen-Steiner2007} (and its generalization the interleaving distance \cite{Chazal2016,Lesnick2015}), and the second is the Wasserstein distance \cite{Cohen-Steiner2010}.
The main idea behind both is as follows.
Given two persistence diagrams, we want to move around points of the first to the configuration of the second with the minimal amount of work.
This is commonly called the earth mover's distance (EMD) \cite{Villani2009,Rubner2000}. The probabilist reading this will immediately note the parallels to the homonymous metric for probability distributions.
While the idea of the Wasserstein distance for persistence diagrams is clearly related to the EMD, the difference is that the EMD allows for mass splitting, while the Wasserstein distance does not.

The Wasserstein distance can be viewed as an $\ell_p$ type metric, as it determines work by summing $p^{\text{th}}$ powers of the distance to move each point; meanwhile the bottleneck is the $\ell_\infty$ analogue as it only considers work to be the farthest distance any point needs to be moved.
The difference in the definition of work leads to differences in computational methods for each (see \cite[Ch.~VIII.4]{Edelsbrunner2010} and \cite{Kerber2016}).
The difference also changes what sort of information is carried in the distance; namely, the Wasserstein distance accounts for all points so is more sensitive to noise, while the bottleneck distance simply sees global structure.

In this paper, we compute the Wasserstein distance for persistence diagrams using a D-Wave quantum computer.
A D-Wave quantum computer uses quantum annealing to solve problems,
a fundamentally different approach from gate-model quantum computers also being explored and developed in industry and academia. 
The gate model machines solve problems expressed in terms of quantum gates, as opposed to polynomials in binary variables.
Recent work has begun to find intersections between TDA and quantum computing, in particular giving methods for calculating Betti numbers using a gate model quantum computer \cite{Lloyd2016,Siopsis2018,Huang2018} and using a D-Wave quantum processor \cite{Dridi2015}.

The D-Wave quantum computer returns minimizing solutions of NP-hard problems formulated as QUBOs~\cite{Garey1979}.
Namely, a QUBO problem is one that minimizes a quadratic polynomial over binary variables.
Traditional computation of the Wasserstein distance requires solving a min-cost matching problem in a constructed bipartite graph \cite{Edelsbrunner2010,Kerber2016}; see \cref{fig:wass-example} for example.
In this paper, we  turn the standard bipartite graph representation into a QUBO (\cref{eq:H}) and show that a minimizing binary solution to this function can be interpreted as the required matching for the Wasserstein computation (\cref{thm:MaxMatchingMinimizesH}).
Finally, we test our algorithm, exploring parameter choices and problem size capabilities, using a D-Wave 2000Q quantum annealing computer. With the rapidly growing capabilities of quantum computers, we expect their performance to yield compelling improvements for tackling complex, and even NP-hard, problems.

\paragraph{Outline.} In \cref{sec:Bkgd}, we give an overview of necessary background for persistent homology, the Wasserstein distance, and the D-Wave quantum computer.
We present our QUBO in \cref{sec:method} and prove correctness in \cref{sec:Equivalence}.
In \cref{sec:experiments} we give the results of our experiments and discuss conclusions in \cref{sec:conclusion}.

%% file: bkgd.tex

\section{Background}
\label{sec:Bkgd}

\subsection{Persistent Homology}

Persistent homology \cite{Edelsbrunner2002,Zomorodian2004,Oudot2017a}, one of the most predominant methods arising from the field of TDA, is based on the following idea.
Given a filtered topological space (likely a simplicial or cubical complex for computational purposes)
\begin{equation*}
\emptyset \subseteq K_1 \subseteq K_2 \subseteq \cdots \subseteq K_n = K,
\end{equation*}
functoriality of homology gives a persistence module
\begin{equation*}
  \begin{tikzcd}
0 \ar[r] & H_*(K_1) \ar[r] & H_*(K_2)  \ar[r] & \cdots \ar[r] & H_*(K_n) = H_*(K).
  \end{tikzcd}
\end{equation*}
This sequence of vector spaces and linear transformations\footnote{We assume homology is computed with coefficients in a field $k$. More often than not, $k = \Z_2$.} can then be studied to understand something about the original filtration.
In general, a (discrete) persistence module $\VV$ is a collection of vector spaces and linear transformations of the form
\begin{equation*}
  \begin{tikzcd}
\VV = ( V_1 \ar[r,"\phi_1"] & V_2 \ar[r, "\phi_2"] & \cdots \ar[r, "\phi_{n-1}"] &V_n)
  \end{tikzcd}
\end{equation*}
with linear transformations $\phi_i^j$ for $i<j$ given by composition $\phi_{j-1} \phi_{j-2} \cdots\phi_i$.
An interval module $I_{[a,b)}$ is a persistence module for which
\begin{equation*}
  V_i =
\begin{cases}
k   & i \in [a,b)\\
0 & \text{else}
\end{cases}
\qquad
\phi_i^j =
\begin{cases}
id   & a \leq i \leq j < b\\
0 & \text{else}.
\end{cases}
\end{equation*}
With reasonable assumptions on the structure,\footnote{See \cite{Bubenik2018} for a complete description.} a persistence module can be decomposed into interval modules
\begin{equation*}
  \VV \cong \bigoplus_{[a,b) \in \BB} I_{[a,b)}
\end{equation*}
where the decomposition may not be unique, but the collection $\BB$ is.
We visualize $\BB$ as a persistence diagram: each $[a,b)$ is drawn as the point $(a,b)$ in the plane $\R^2$.
As we always have $a < b$, we include the diagonal $\Delta = \{(c,c) \mid c \in \R\}$ when drawing a persistence diagram.
We say that a point $[a,b)$ in a persistence diagram represents a feature that is born at $a$, dies entering $b$, and has lifetime $b-a$.

Points in the persistence diagram far from the diagonal represent a homology class that appeared early in the filtration, and stayed for a long time relative to the length of the filtration.
For example, in \cref{fig:ExamplePersistence}, we start with a point cloud $P$ embedded in $\R^2$ and let $K_i$ be the union of balls of radius $r_i$, $\bigcup_{x \in P} B_{r_i}(x)$.%
\footnote{This is a mild simplification. In reality, computation is done with the Vietoris-Rips complex, a simplicial complex that has approximately the same topology as the union of disks.}
The 1-dimensional homology, which measures circular structures, is used to construct the persistence module, and the resulting persistence diagram is shown at right.
The two circular pieces of the point cloud are encoded in the two points in the persistence diagram that are far from the diagonal.
While there is an ongoing debate as to the ``right'' way to interpret importance of points far from the diagonal, it is clear that having a point close to the diagonal should be almost as if that point were not included at all (that is, if it became a degenerate interval $[a,a)$).
The metric used for persistence diagrams takes this interpretation into consideration.

\begin{figure}
  \includegraphics[width = .19\textwidth]{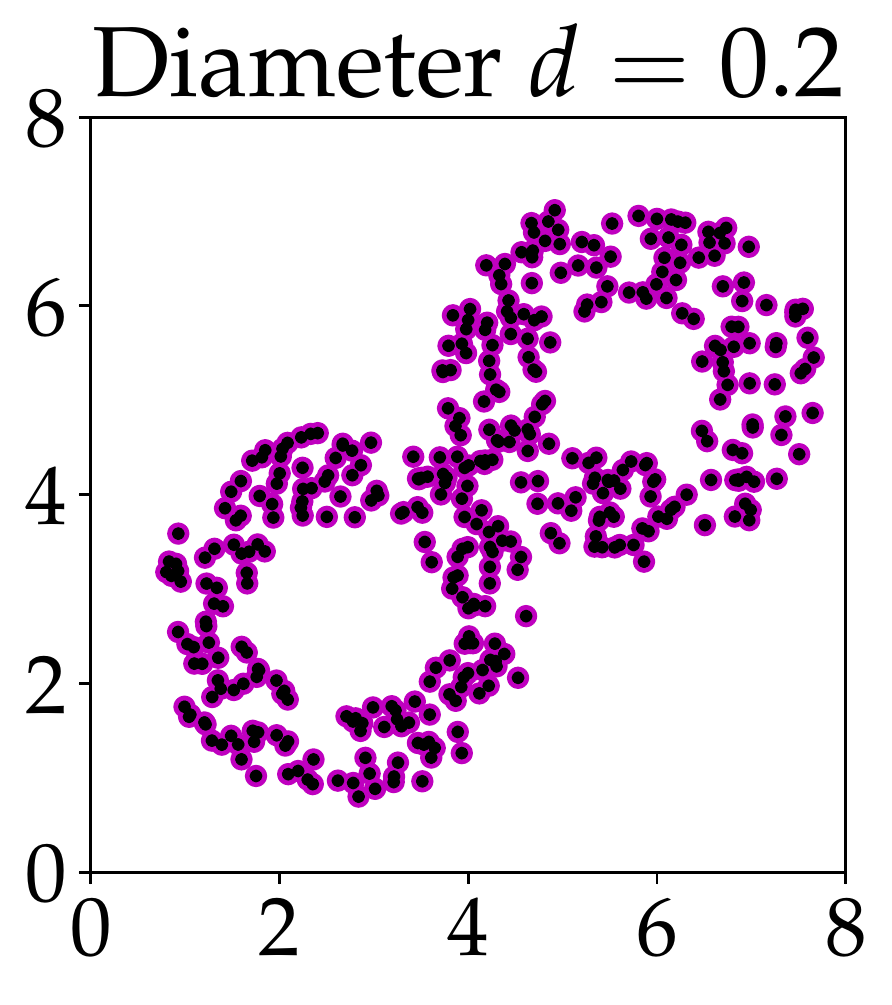}
  \includegraphics[width = .19\textwidth]{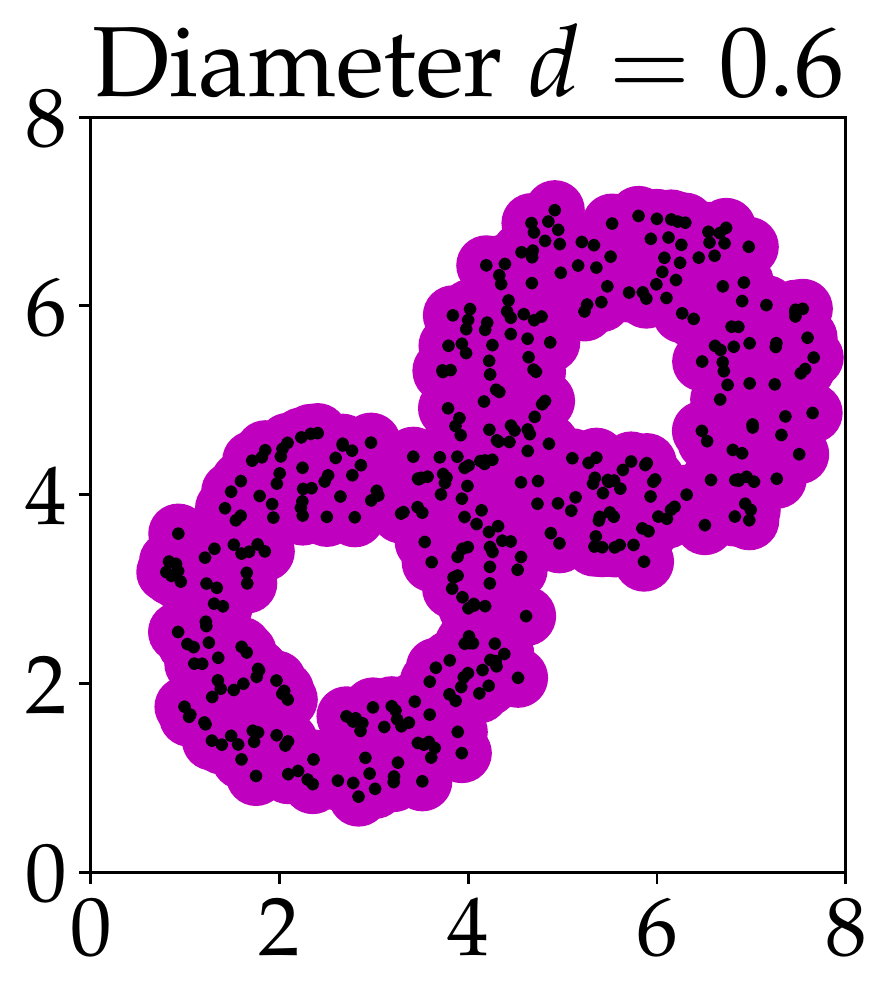}
  \includegraphics[width = .19\textwidth]{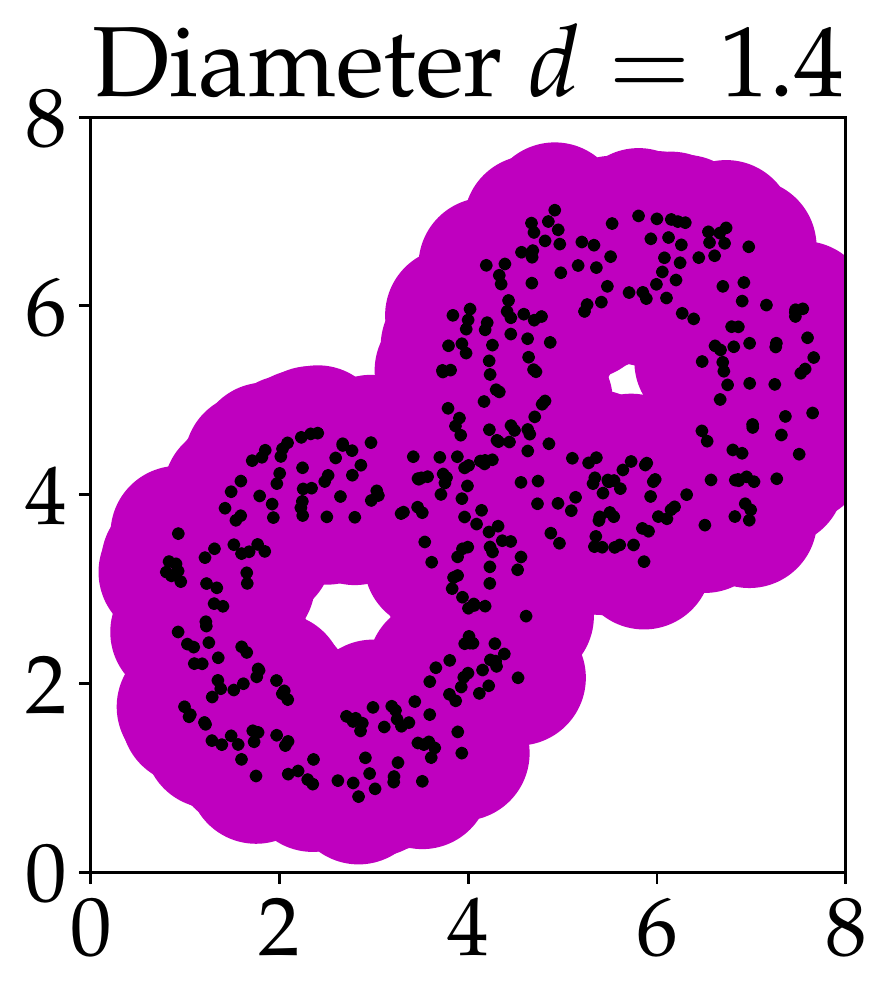}
  \includegraphics[width = .19\textwidth]{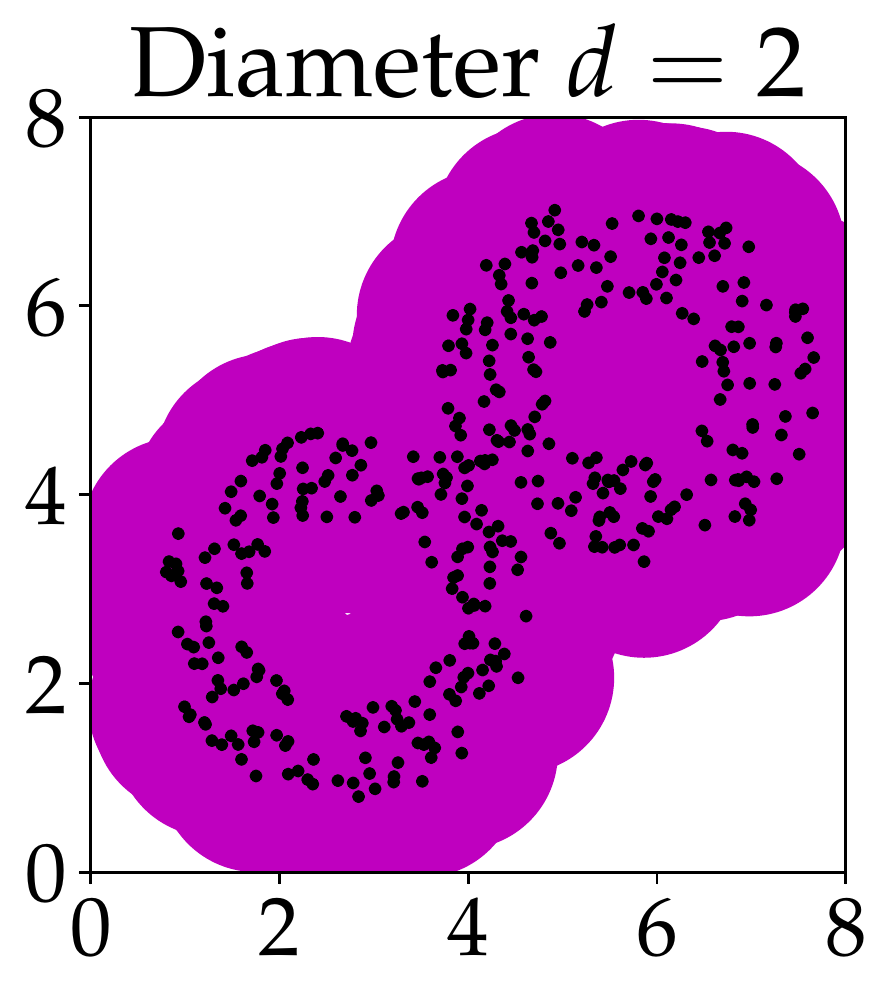}
  \includegraphics[width = .2\textwidth]{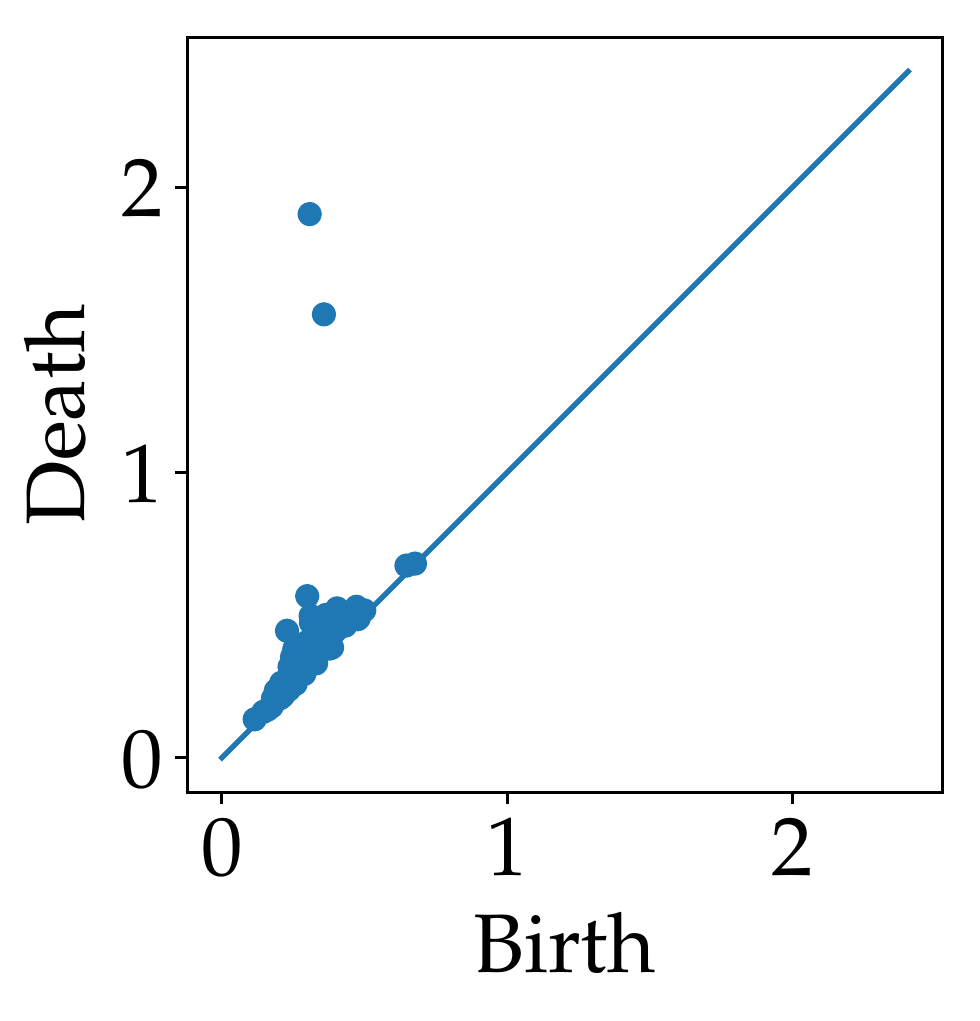}
  \caption{An example of the persistence diagram (right) for a filtration given by the union of disks with increasing radius.  The two points in the diagram far from the diagonal indicate that the point cloud seems to be drawn from a double annulus.}
  \label{fig:ExamplePersistence}
\end{figure}

\subsection{Wasserstein distance for persistence diagrams}

In theory, a persistence diagram is a finite multiset of points in $\R^2$ above the diagonal $\Delta = \{(c,c) \in \R^2\}$, along with a countably infinite set of copies of each point on the diagonal.
In practice, a persistence diagram is simply represented by the off-diagonal points as follows.
Let $\R^2_{>\Delta}$ be the portion of the plane above the diagonal, $\R^2_{>\Delta}=\{(c,d) \mid d > c\}$.
We will notate a diagram as $X = \{a_1,\cdots,a_n\}$ for $a_i \in \R^2_{> \Delta}$.

When considering a metric for these objects, it is important to take the meaning of the points into account; namely, that a point close to the diagonal $(c,c+\e)$ represents a feature that lived for a short time, $\e$.
A diagram with this small lifetime point should intuitively be close to the same diagram without that point, as if the feature had never appeared at all.
Hence, we consider the distance between two diagrams to be the minimal cost required to match up their points, either matching off-diagonal to off-diagonal, or off-diagonal to the nearest point on the diagonal, with respect to some cost function for the matching.
There are two closely related options for a metric on persistence diagrams: the bottleneck distance, and the $p^\text{th}$ Wasserstein distance ($p \geq 1$).
The main difference between them is simply the choice of cost function for a matching.
In this paper, we focus on the Wasserstein distance.

\begin{dfn}
The $p^{\text{th}}$ Wasserstein distance is defined to be
\begin{equation*}
d_p(X,Y) = \inf_{\phi:X \to Y} \left( \sum_{a \in X} \|a - \phi(a)\|_q^p \right)^{1/p}
\end{equation*}
where the infimum is taken over all bijections between $X$ and $Y$.
\end{dfn}
For historical reasons, $q=\infty$ in the majority of applications.
It is often more reasonable to set $q = p$ to control the geometry of the space of persistence diagrams \cite{Turner2014}, but we need no such assumption here.
We call
\begin{equation}
  \label{eq:WassCostBijection}
  \sum_{a \in X} \|a - \phi(a)\|^p_q
\end{equation}
 the $p$th Wasserstein cost function for the bijection $\phi$.

One should get mathematically nervous over where this infinite sum converges, however, it is easy to see that we always have a bijection with finite cost.
Namely, we associate all the off-diagonal points with their projection to the diagonal in the other diagram.
Then all but finitely many of the points on the diagonal $\Delta$ of one diagram will be matched to their counterparts in the diagonal of the other diagram, possibly with a Hilbert-hotel-style shift in index if a copy of that diagonal point had been matched to an off-diagonal point.
This may not be the best bijection, however, we can at least determine that $d_p(X,Y) < \infty$.

\subsection{Existing methods for computation}

From a computational perspective, one can provide the equivalent definition of Wasserstein distance through a best matching in a bipartite graph as follows.
First, recall that a \textit{matching} in a bipartite graph is a collection of edges $\x \subseteq E(G)$ such that each vertex is adjacent to at most one edge in $\x$; a \textit{perfect} matching is a matching such that each vertex is adjacent to exactly one edge in $\x$.
A \textit{maximal} matching is a matching that is not properly contained in a larger matching; that is, a matching $\x \subseteq E$ is maximal if there does not exist a matching $\y \subseteq E$ with $ \x \subsetneq \y$.
Note that in the case of a complete bipartite graph where the sizes of the vertex sets are the same, any maximal matching must be a perfect matching.

As a first attempt to compute the distance between diagrams $X = \{a_1,\cdots,a_n\}$ and $Y = \{ b_1,\cdots,b_m\}$, we construct the following complete  bipartite graph.
For the sake of notation, denote the projection of $a \in \R^2_{>\Delta}$ to the diagonal by $\Delta_a$; for $a = (c,d)$, this is $\Delta_a = ((c+d)/2, (c+d)/2)$.
For a set $S \subset \R^2_{>\Delta}$, let $\Delta_S = \{\Delta_a \mid a \in S\}$.
The weighted bipartite graph $G$ we build has vertex set $U \sqcup V$ with  $U = X \sqcup \Delta_Y$ and $V = Y \cup \Delta_X$.
Note that $|U| = |V| = n+m$.
We write $u \in X$ or $v \in Y$ for vertices representing off-diagonal points and $u \in \Delta$ for vertices representing points on the diagonal.
We will also often abuse notation by using $u$ or $v$ interchangeably to represent the vertex in the graph or the point in the diagram.

The weight for each edge $(u,v)$ is defined to be $\omega(u,v) = \|u-v\|_q^p$, the $p^\text{th}$ power of the distance  between the represented points in $\R^2$.
One particularly useful calculation is that the distance between a point and its projection to the diagonal, i.e.,~the weight of an edge of the form $(a,\Delta_a)$ for $a = (c,d)$ is $\left(\frac{d-c}{2^{1-1/q}}\right)^p$.
In the case that $q=\infty$, this becomes $\left(\frac{d-c}{2}\right)^p$.

The cost of a matching $\x \subseteq E$ (c.f.~\cref{eq:WassCostBijection}) is defined to be
\begin{equation}\label{eq:cost-p}
    C_p(\x) = \sum_{e \in \x} \omega(e) = \sum_{(u,v) \in \x} \|u-v\|_q^p.
\end{equation}
A minimum cost maximal matching is a maximal matching $\x$ such that $C(\x) \leq C(\x')$ for all maximal matchings $\x'$.
Note that because we are working with a complete bipartite graph, any maximal matching is also a perfect matching, so this could have been defined via perfect matchings; however, we will need this generality later.
By the Reduction Lemma \cite[Sec.~VIII.4]{Edelsbrunner2010}, if $C(\x)$ is the minimum cost of maximal matchings of the graph $G$, then $C(\x)^{1/p}$ is  the $p^{\text{th}}$ Wasserstein distance between the diagrams $X$ and $Y$.

There is a trick that can be employed for computation, namely to work with a smaller bipartite graph  than $G$.
Let $\widetilde G:= \widetilde G(X,Y)$ be a bipartite graph on the same sets $U$ and $V$ for which $\widetilde G |_{X \cup Y}$ is a complete bipartite graph, and the remaining edges are of the form $(z,\Delta_z)$ for $z \in X \cup Y$.
This means that $|E(\widetilde G)| = nm + n + m$.
We are no longer interested in perfect matchings, since the only such available matching in $\widetilde G$ is the one consisting of edges $\{(z, \Delta_z) \mid z \in X \cup Y\}$.
We do need a property of maximal matchings in $\tilde G$, which will be useful later.
\begin{lem}
    \label{lem:UnmatchedVertexInTildeG}
Any unmatched vertex in a maximal matching in $\widetilde G$ is $\Delta_z$ for some $z$.
\end{lem}
\begin{proof}
Let $\x \subseteq E$ be a maximal matching in $\tilde G$ where vertex $u \not \in \Delta$ is unmatched.
But then $\Delta_u$ is also unmatched as it has degree 1, so $\x$ can be increased by including the edge $(u,\Delta_u)$, contradicting maximality.
\end{proof}

We still look for a min-cost maximal matching in this reduced graph and use it to determine the Wasserstein distance.
That it is reasonable to use $\widetilde G$ instead of $G$ can be seen in the following lemma, which is an immediate consequence of \cite[Lem.~2.2]{Kerber2016}.

\begin{lem}[{\cite{Kerber2016}}]
  \label{lem:SmallerBipartiteGraph}
If $C(\x)$ is the cost of a minimum cost maximal matching in $\widetilde G$, then
$$d_p(X,Y) = C(\x)^{1/p}.$$
\end{lem}

In particular, \cref{lem:SmallerBipartiteGraph} implies that we can work with the smaller graph $\tilde G$, resulting in a size reduction of about half in practice \cite{Kerber2016}.
Once we have constructed the bipartite graph $G$ or $\widetilde G$, we can use any standard min-cost matching algorithm, i.e., the Hungarian algorithm \cite{Kuhn1955}, to find a min-cost maximal matching (MCMM).
In practice, the current state-of-the-art software is Hera \cite{Kerber2016}, which uses the fact that the points in the diagram are points in the plane in order to speed up computation of the Wasserstein distance.



\subsection{QUBOs and Quantum Annealing}
\label{sec:qubos}

The D-Wave quantum computer seeks minimum energy solutions to
 a combinatorial optimization problem known as a QUBO. This problem is mathematically
equivalent to the Ising problem in statistical mechanics, first posed by
Wilhelm Lenz in 1920 and solved in one dimension by Ernst Ising in his
1924 Ph.D.~thesis \cite{Ising1925,Pathria2012}. 
See \cref{tab:qubo-terms} for a list of terminology used in the field.

The D-Wave programmable quantum computer (QC) is comprised of a grid of
superconducting loops, each of which acts as a programmable flux qubit \cite{Harris2010}.
The loops may have current in either direction, corresponding to
up and down spins. The individual qubits are arranged in a grid
with couplers corresponding to controllable mutual inductances
between the magnetic fields associated with the current loops.
The grid is a bipartite graph also known as Chimera.

\begin{table}[tb]
  \centering
  \begin{tabular}{|l|p{11cm}|}
  \hline
  Term           & Definition \\
  \hline
  qubit          & {\em Quantum bit} that participates in annealing cycle and settles into one of two possible final states: \{0,1\} \\
  coupler        &Physical device that allows one {\em qubit to influence another qubit} \\
  weight or bias &  Real-valued {\em constant associated with each qubit}, which influences the qubit’s tendency to collapse into its two possible final states; controlled by the programmer \\
  strength       & Real-valued {\em constant associated with each coupler}, which controls the influence exerted by one qubit on another; controlled by the programmer \\
  objective      &Real-valued {\em function which is minimized} during the annealing cycle \\
  \hline
  \end{tabular}
  \caption{Terminology for quantum computing.}
  \label{tab:qubo-terms}
\end{table}

The quantum computer implements a process known as quantum annealing \cite{Kadowaki1998,Farhi2001}.
The system starts in a state described by the initial
quantum-mechanical Hamiltonian whose lowest energy state, or ground
state, is a superposition of all possible computational basis states.
The goal of the annealing process is to find the ground state
of a final Hamiltonian, specified by the user. The system evolves according
to the time-dependent Schr{\"o}dinger equation. The adiabatic theorem
\cite{Born1928}
says that if the time evolution of the system is slow
enough, then the system remains in its ground state.
Therefore, at the end of a slow annealing process, the final state will be the
ground state of the input Hamiltonian, which is also a global minimum of the
objective function.

During the process, the system samples from an approximate Boltzmann
distribution over the energy landscape defined by the Hamiltonian, the specific form of which is problem-dependent.
The system can be sampled many times once the problem has been translated onto the hardware.
An important aspect is that the system naturally samples from a probability distribution, and a user can obtain hundreds or thousands of samples in order to explore that distribution.
This will be seen, in particular, in our experiments where the QC returns distributions of solutions, rather than a single answer (see \cref{fig:wass-dists-all}).

A problem is first expressed in terms of binary variables with real coefficients where the highest power of variable appearing is two.
Expressing the problem as a QUBO often includes writing out constraints; e.g., turn on one and only one variable in a column in a matrix.
A 2014 paper by Andrew Lucas showed QUBO formulations for the 21 NP-hard problems explored by Richard Karp \cite{Lucas2014,Karp1972,Garey1979}.
QUBOs are found in many fields, including portfolio management, job-shop scheduling, and traffic engineering.

When problems require more connectivity than is available in the bipartite
Chimera graph, it is often possible to embed them onto the graph by forming
chains of qubits, in which the qubits are constrained to have the same value \cite{Cai2014,Choi2008a, Choi2010}. 
This idea will also be important to understand the results of our experiments in \cref{sec:experiments}.

%% file: method.tex

\section{Method}
\label{sec:method}

We will turn our bipartite graph $\widetilde G$ into a QUBO so that the solution can be interpreted as a matching between the graphs, and the value of the minimized solution is equal to the $p^\text{th}$ power of the Wasserstein distance.
In particular, we need to build a QUBO that is minimized exactly when the answer both represents a MCMM while minimizing the sum of the edge weights.

Given diagrams $X = \{a_1,\cdots,a_n\}$ and $Y = \{ b_1,\cdots,b_m\}$, we have the bipartite graph $\widetilde G$ with vertex sets $U \cong X \cup \Delta_Y$ and $V \cong Y \cup \Delta_X$, so that $|U| = |V| = n+m$.
For $(u,v) \in E:= E(\widetilde G)$ adjacent, set $\omega(u,v) = \|u - v\|_q^p$.
Set
\begin{equation}
  \label{eq:DefB}
   B > B^* := \max_{(u,v) \in E(\widetilde G)} \omega(u,v) .
\end{equation}

We assume genericity of the points in $X$ and $Y$; that is, $\|a-b\| >0$ for all off-diagonal points $a \in X, b \in Y$.

We will build a QUBO on $M:= nm + n + m$ variables $\x = \{x_{u,v} \mid (u,v) \in E \}$.
First, note that values of $\x \in (\Z_2)^{nm + n + m}=: Z$ are in bijection with subsets of edges, $\{e \in E \mid x_{u,w} = 1\}$, so we abuse notation and write $\x \in Z$ or $\x \subseteq E$ depending on the context.
We also write, e.g., $u \in X \subseteq U$ for vertices associated with the off-diagonal points of the diagram $X$ and $u \in \Delta_Y \subseteq U$ for the vertices associated with the diagonal points of the diagram $Y$.
Set
\begin{align}
F_c(\mathbf{x})  & =  \sum_{(u,v) \in E} \omega(u,v) x_{u,v} \label{eq:Fc}  \\
F_{U}(\mathbf{x}) & = B \sum_{u \in X \subset U} \left( 1-\sum_{\substack{v \in V \\ (u,v) \in E}} x_{u,v}\right)^2 \label{eq:Fu} \\
F_V(\mathbf{x}) & = B \sum_{v \in Y \subset V} \left( 1-\sum_{\substack{u \in U \\ (u,v) \in E}} x_{u,v}\right)^2 \label{eq:Fv}
\end{align}
where $B$ is a non-negative, real-valued Lagrangian multiplier. Then we are interested in the QUBO
\begin{align}\label{eq:H}
    H = F_c + F_U + F_V.
\end{align}

As long as $\x$ represents a matching, $F_c(\x)$ is built to return the cost of the matching.
To see this, note that fixing $u \in X \subseteq U$, $\sum_{\substack{w \in W \\ (u,w) \in E}} x_{u,w}$ is simply the number of edges in $\x$ adjacent to $u$.
Thus, the term
$$\left( 1-\sum_{\substack{v \in V \\ (u,v) \in E}} x_{u,v}\right)^2$$
 is zero if vertex $u$ is adjacent to exactly one edge in $\x$, and strictly positive otherwise.
 If $\x$ is a maximal matching, and thus by \cref{lem:UnmatchedVertexInTildeG}, every non-diagonal vertex is adjacent to exactly one edge, then  $F_U$ and $F_V$ are built  so that  $F_U(\x) = F_V(\x) = 0$.
In particular, this means that for a maximal matching, $H(\x) = C(\x)$.

\subsection{An example}
\label{sec:example}

\begin{figure}
  \centering
  \includegraphics[width = .6\textwidth]{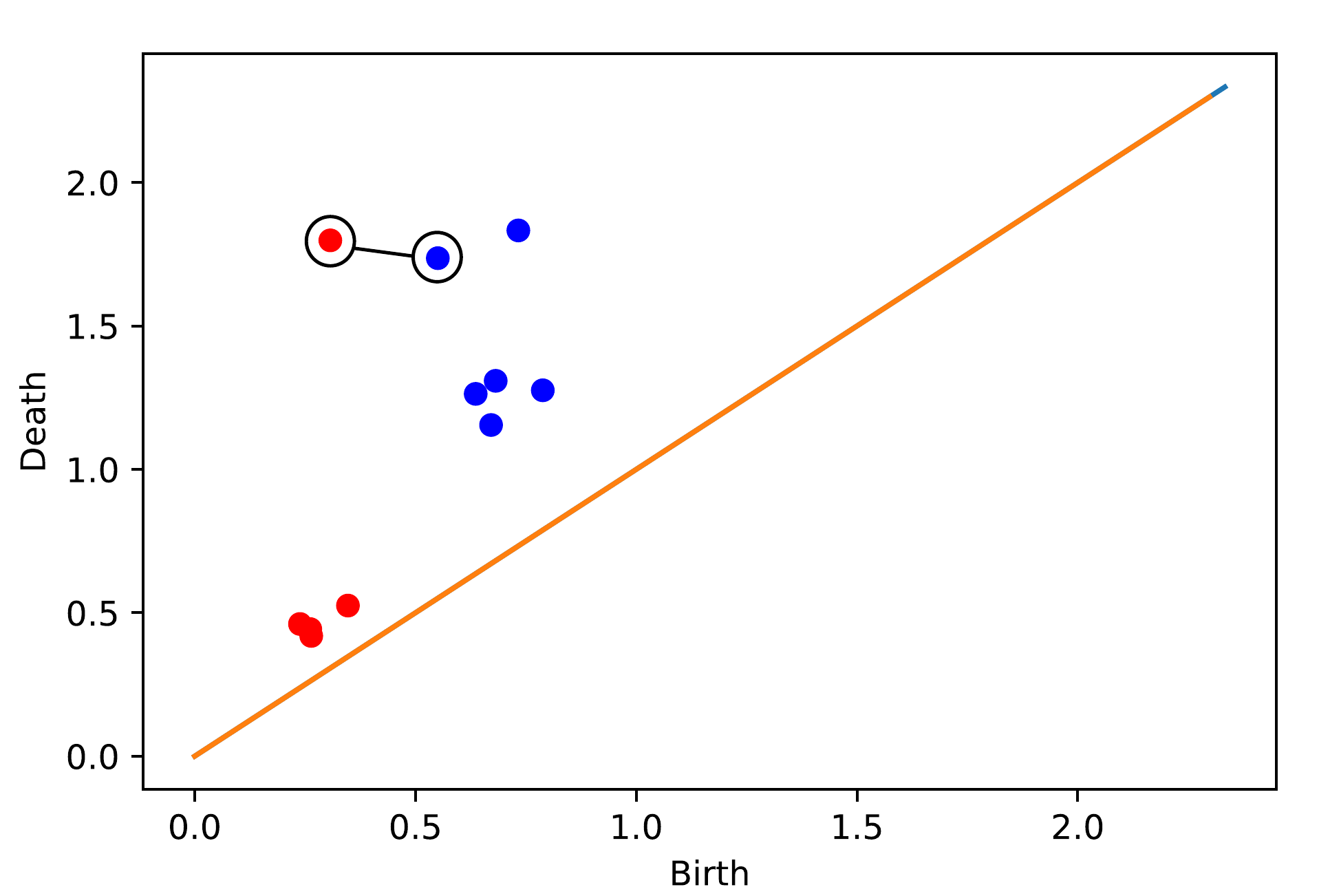}
  \caption{Two persistence diagrams overlaid, one from points sampled from a torus with two long-lived generators ({\color{blue}blue}) and the other from points sampled from an annulus with a single long-lived generator ({\color{red}red}). Some noisy (short-lived) generators have been truncated to simplify the example. The circled points and the edge between them highlights the pairing shown in \cref{fig:wass-graph-full} between $(5,1)$.}
  \label{fig:dgm-example}
\end{figure}

We present a brief example detailing the formulation of a bipartite graph from two persistence diagrams and construction of the resulting QUBO. Consider \cref{fig:dgm-example} showing two diagrams overlaid. The diagrams correspond to a noisy torus and a noisy annulus and have been truncated to reduce short-lived generators near the diagonal. In the example in \cref{fig:dgm-example}, diagram $X$ contains $n=6$ points (blue) and diagram $D_2$ contains $m=5$ points (red).

We first construct $\tilde G$ for this example in \cref{fig:wass-graph-full}.
On the left-hand side of $\tilde G$, blue circles (\coloredcircle[blue, fill=blue]{3pt}) represent the $n$ points from $X$.
On the right hand side, the red squares (\coloredbox{red}) correspond to the $m$ points from $Y$.
Each of $n$ nodes from $X$ has an edge to the opposing $m$ nodes from $Y$, and vice versa.
These pairings are represented by the edges (\coloredcircle[blue, fill=blue]{3pt}, \coloredbox{red}) in \cref{fig:wass-graph-full}.
In addition to edges between off-diagonal points, there are edges between each point and its projection onto the diagonal.
Nodes for diagonal points are situated on the opposing side and are indicated by $\Delta_{*}$'s.
For instance, vertex 0 in $X$, i.e., \coloredcircle[blue, fill=blue]{3pt} labeled 0, is paired to $\Delta_0$ on the right hand side of \cref{fig:wass-graph-full}.
The graph is not a fully-connected bipartite graph since $\Delta_{*}$'s pair only with their off-diagonal representative.
Thus, the graph $\widetilde{G}(X,Y)$ is composed of nodes $X = \{\coloredcircle[blue, fill=blue]{3pt}\} \cup \{{\color{blue}\blacktriangle}\}$ and $Y = \{\coloredbox{red}\} \cup \{{\color{red}\blacktriangle}\}$, with full connectivity between nodes representing off-diagonal and single connectivity between off-diagonal nodes and their diagonal projection representatives.

In \cref{fig:wass-graph-full}, one possible maximal matching is shown in bold.
Note that every non-diagonal node has degree exactly 1 in the matching; only diagonal nodes need not be paired as per \cref{lem:UnmatchedVertexInTildeG}.
The pairing of 5 and 0 in \cref{fig:wass-graph-full} is highlighted by the black circles connected by an edge in \cref{fig:dgm-example}; all other points are paired to their diagonal projection.
The Wasserstein distance for the matching is then the $\sum_{(u,v)}\omega(u,v)$ for all bold edges $(u,v)$.

From this graph, we construct the QUBO drawn in \cref{fig:QUBO-coefficients}.
In this matrix, each row represents the variable for one edge $x_{u,v}$ in the graph.
The color of the entry $(x_{u,v},x_{u',v'})$ in the upper triangular portion represents the coefficient of the monomial $x_{u,v}x_{u',v'}$.
Due to lexicographic sorting of variables in the code, these are sorted in the matrix as follows.
Denote the vertices on the left by $u_i$ and those on the right by $v_i$. Containment of vertex $\Delta_i$ is clear from context despite being duplicated. The resulting order of edges is
\begin{align*}
  (u_0,v_0), (u_0,v_1),&\cdots,(u_0,v_4),(u_0, \Delta_0), \\
  (u_1,v_0), (u_1,v_1)&\cdots, (u_1,v_4),(u_1,\Delta_1),\\
  &\phantom{x}\vdots\\
  (u_5,v_0), (u_5,v_1)&\cdots, (u_5,v_4), (u_5,\Delta_5),\\
  (\Delta_0,v_0), (\Delta_1&,v_1), \cdots (\Delta_4,v_4).
\end{align*}

\begin{figure}
  \centering
  \begin{subfigure}{0.45\textwidth}
    \includegraphics[width=\textwidth]{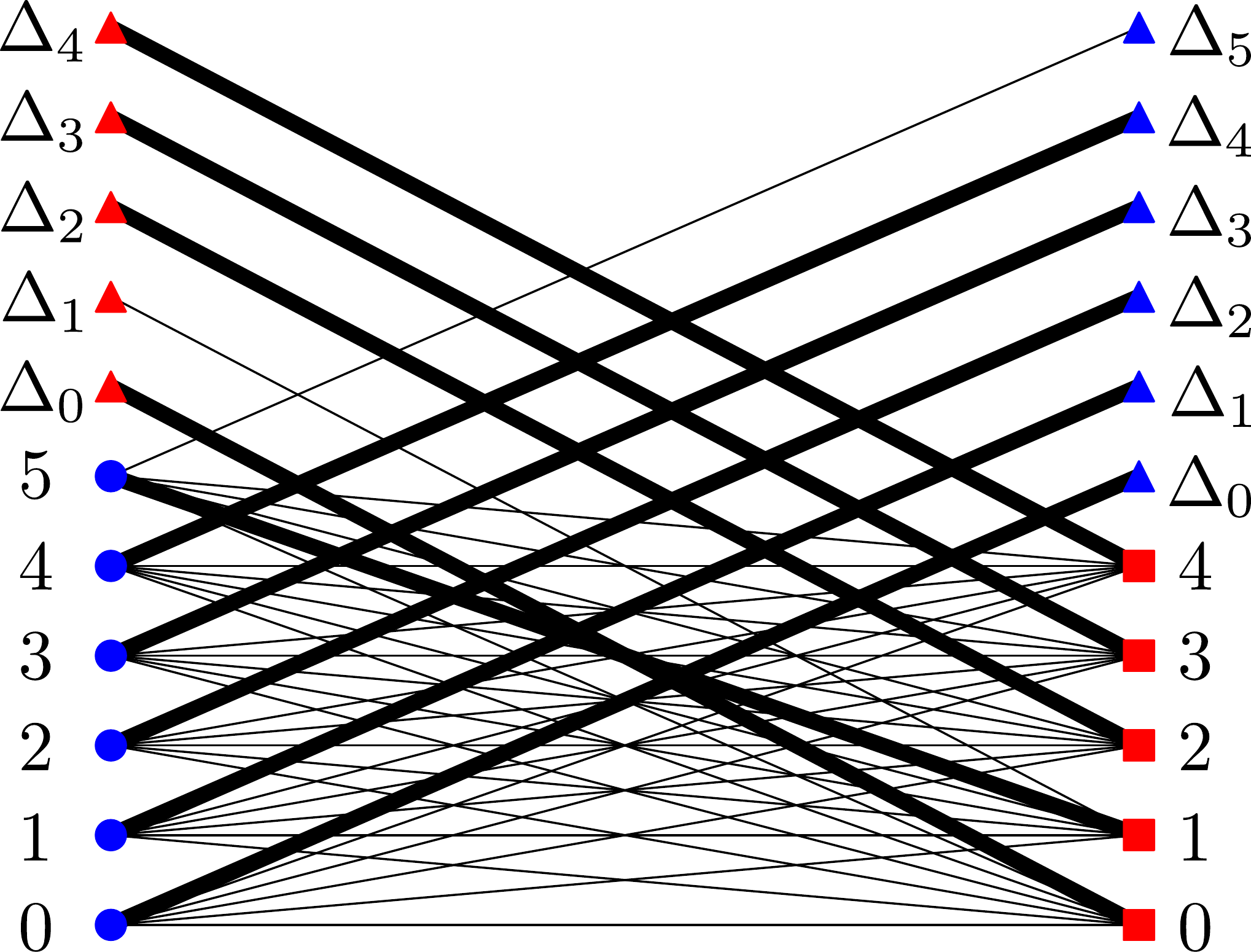}
    \caption{}
    \label{fig:wass-graph-full}
  \end{subfigure}
  \begin{subfigure}{0.45\textwidth}
    \includegraphics[width=\textwidth]{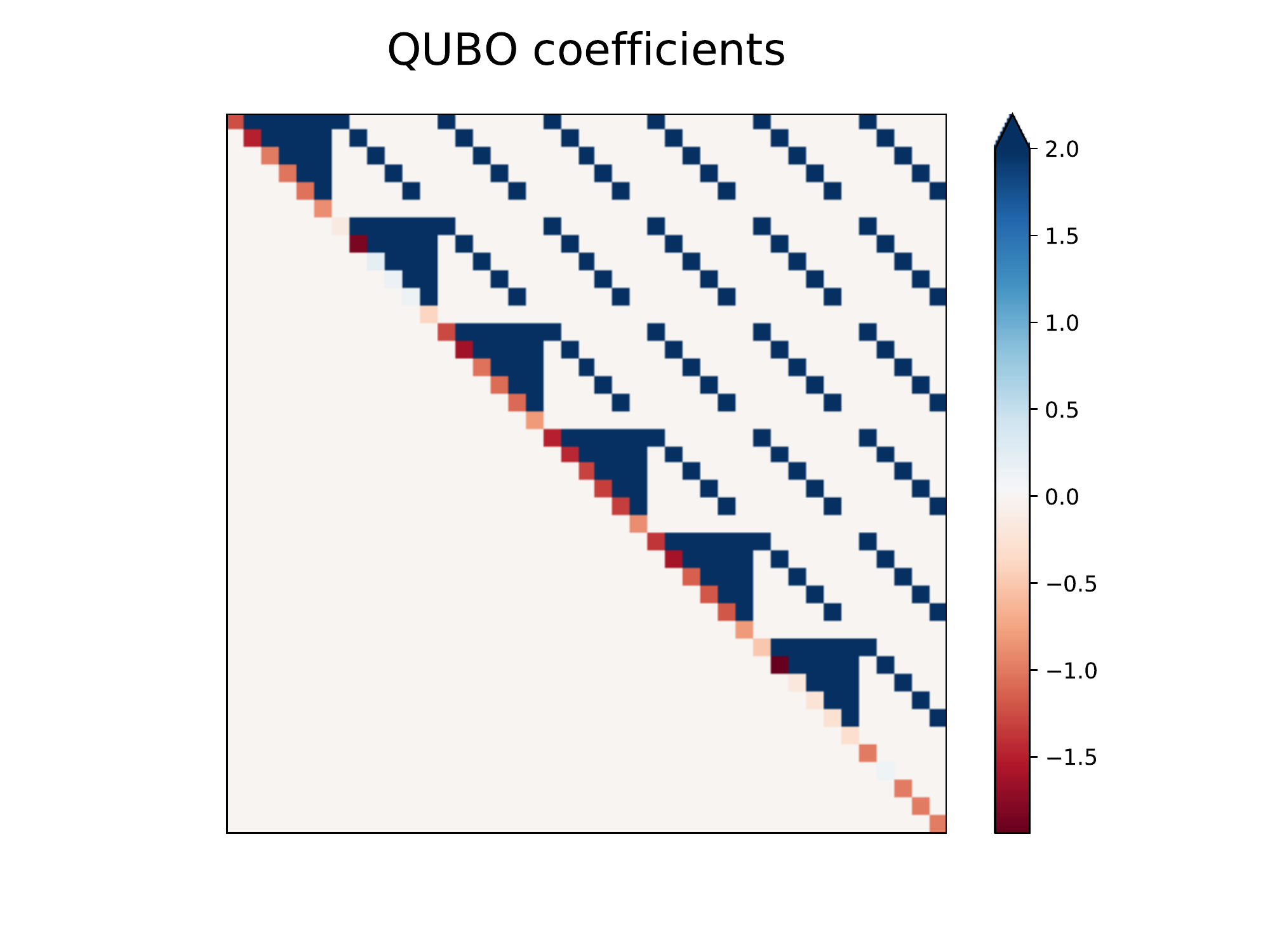}
    \caption{}
    \label{fig:QUBO-coefficients}
  \end{subfigure}
  \caption{
  Left, the graph $\tilde G$ for the example persistence diagrams of \cref{fig:dgm-example}.
  Blue circles and triangles correspond to points from $D_1$ and their diagonal pairings; red squares and triangles to points and diagonals pairings in $D_2$.
  The weights on the edges (not drawn) are the distance between the representative points.
  The bold edges are those chosen in a MCMM.
  For this example, most edges connected points to their diagonals. Only one of the long-lived generators of the torus (point 5) is matched to the point representing the robust generator of the annulus (point 1).
  At right, the graph representing the coefficients of the constructed QUBO.
  Each row represents one of the variables $x_{u,v}$, and color in the upper half plane for entry at $(x_{u,v}, x_{u'v'})$ represents the the coefficient of the monomial $x_{u,v}x_{u'v'}$.}
  \label{fig:wass-example}
\end{figure}

%% file: equivalence.tex

\section{Equivalence}
\label{sec:Equivalence}

In this section, we show that the QUBO built in \cref{sec:method} is minimized exactly when the binary variables can be read off as an MCMM for $\tilde G$.

\begin{thm}
  \label{thm:MaxMatchingMinimizesH}
Let $\mathbf{x} \in Z$.
Then $\mathbf{x}$  is a solution which minimizes $H$ if and only if
$\x \subseteq E$
is a MCMM of $\widetilde G$.
\end{thm}

Before we can prove the theorem, we will need a few technical lemmas.

\subsection{A lemma for calculations}
\begin{lem}
    \label{lem:CalculationRules}
Assume $\y, \tilde \y \in Z$ are equal except for a single entry $\y_{u,v} \neq \tilde \y_{u,v}$.
Let $(a, T) \in \{ (u,U), (v,V) \}$.
Then we have the following table for calculations.
\begin{center}
\begin{tabular}{|cc|c|c||c|c|}
    \hline
    $\y_{u,v}$ & $\tilde \y_{u,v}$ &
        $a$ &
        No.~neighbors of $a$ in $\y$ &
        Result &  \\ \hline \hline
     1 & 0 & &  & $F_c(\tilde \y) = F_c(\y) - \omega(u,v)$ & (a)\\
     0 & 1 &&   & $F_c(\tilde \y) = F_c(\y) + \omega(u,v)$ & (b)\\\hline \hline
    & & $ a \in \Delta \subseteq T$ & & $F_T(\tilde \y) = F_T(\y)$  & (c)\\ \hline
    0 & 1 & $ a \not \in \Delta \subseteq T$ & 0 & $F_T(\tilde \y) = F_T(\y) - B$ & (d)\\
    0 & 1 & $ a \not \in \Delta \subseteq T$ & $\geq 1$ & $F_T(\tilde \y) \geq F_T(\y) + B$ & (e)\\ \hline
    1 & 0 & $ a \not \in \Delta \subseteq T$ & $1$ & $F_T(\tilde \y) = F_T(\y) + B$ & (f)\\
    1 & 0 & $ a \not \in \Delta \subseteq T$ & $\geq 2$ & $F_T(\tilde \y) \leq F_T(\y) - B$ & (g)\\ \hline
\end{tabular}
\end{center}

\end{lem}

\begin{proof}
Rows (a) and (b) are immediate from the fact that $\tilde \y$ and $\y$ differ by exactly one entry.
For the remainder of the rows, WLOG, we will assume $(a,T) = (v,V)$ as the case for $(u,U)$ is symmetric.

If $v \in \Delta_X \subseteq V$, then because the sum in $F_V$ ignores the edges adjacent to $\Delta$, there is no difference between $F_V(\tilde \y)$ and $F_V(\y)$, which proves row (c).
Thus, assume $v \not \in \Delta \subseteq V$.
If $\y_{u,v} = 0$ and $\tilde \y_{u,v} = 1$, this constitutes adding an edge.
If $v$ had no adjacent edges in $\y$, then
\begin{equation*}
1-\sum_{\substack{u' \in U \\ (u',v) \in E}} y_{u',v} = 1,\qquad \text{ and}\qquad
1-\sum_{\substack{u' \in U \\ (u',v) \in E}} \tilde y_{u',v} = 0,
\end{equation*}
which gives row (d).
If $w$ has $k \geq 1$ adjacent edges in $\y$, then it has $k+1$ adjacent edges in $\tilde \y$.
So
\begin{align*}
    F_V(\tilde \y) & = F_V(\y) - (1-k)^2 B + (1-(k+1))^2 B \\
    & = F_V(\y) + (2k-1))B\\
    & \geq F_V(\y) + B,
\end{align*}
which proves row (e).

Assume  that, still, $v \not \in \Delta \subseteq V$, but ve are instead removing an edge, so $\y_{u,v} = 1$ and $\tilde \y_{u,v} = 0$.
If $v$ has exactly one neighbor in $\y$, so none in $\tilde y$,
\begin{equation*}
1-\sum_{\substack{u' \in U \\ (u',v) \in E}} y_{u',v} = 0, \text{ and}\qquad
1-\sum_{\substack{u' \in U \\ (u',v) \in E}} \tilde y_{u',v} = 1,
\end{equation*}
thus $F_V(\tilde \y) = F_V(\y) + B$ (row (f)).
On the other hand, if $v$ has $k \geq 2$ adjacent edges in $\y$, then
\begin{align*}
    F_V(\tilde \y) & = F_V(\y) - (1-k)^2 B + (1-(k-1))^2 B \\
    & = F_V(\y) + (3-2k)B\\
    & \leq F_V(\y) -B,
\end{align*}
which is needed for row (g).
\end{proof}

\subsection{From non-matching to matching}
The next  two results, \cref{lem:removeOneEdge} and \cref{prop:NonMatchingVsMatching}, show that a given $\x$, which is not a matching, can be turned into a subset $\y$, which is a matching and has a strictly lower value of $H$.

\begin{lem}
    \label{lem:removeOneEdge}
If $\x \subseteq E$ is not a matching, and $\tilde \x$ differs from $\x$ by only removing an edge of $\x$ adjacent to at least one other edge in $\x$, then $H(\tilde \x) < H(\x)$.
\end{lem}

\begin{proof}
Assume $\x$ is not a matching, and without loss of generality, assume that node $v \in V$ has two or more neighbors in $\x$.
Note that by assumption, this means that $v \not \in \Delta_X \subset V$, so $v \in Y \subset V$.
Let $\tilde \x$ differ from $\x$ by exactly one edge: $y_{u,v} = 1$ and $\tilde y_{u,v} = 0$.

First, using \cref{lem:CalculationRules}(a), we know that $F_c(\tilde \x) = F_c(\x) - \omega(u,v)$, so we need check only cases where $F_U$ and $F_V$ differ.
If $u = \Delta_v$, then using \cref{lem:CalculationRules} rows (c) and (g), we have
\begin{equation*}
H(\tilde \x) \leq
    \big(F_c(\x) - \omega(u,v)\big)
    + \big(F_U(\x)\big)
    + \big(F_V(\x) - B\big)
    < H(\x)
\end{equation*}
as $B > \omega(u,v)$.
If $u \in X \subseteq U$ has exactly one adjacent edge in $\x$,
\begin{equation*}
H(\tilde \x) \leq
    \big(F_c(\x) - \omega(u,v)\big)
    + \big(F_U(\x) + B\big)
    + \big(F_V(\x) - B\big)
    < H(\x)
\end{equation*}
as $\omega(u,v) > 0$.
Finally, if $u \in X \subseteq U$ has at least $ 2$ adjacent edges in $\x$, then by \cref{lem:CalculationRules}(g),
\begin{equation*}
H(\tilde \x) \leq
    \big(F_c(\x) - \omega(u,v)\big)
    + \big(F_U(\x) - B\big)
    + \big(F_V(\x) - B\big)
    < H(\x).
\end{equation*}

\end{proof}


\begin{prop}
\label{prop:NonMatchingVsMatching}
If $\x \subseteq E$ is not a matching, then there exists $\y \subset \x$ which is a matching such that $H(\x) > H(\y)$.
\end{prop}

\begin{proof}
  If $\x$ is not a matching, then there must be a vertex adjacent to more than one edge in $\x$.
  By \cref{lem:removeOneEdge}, removing this edge from $\x$ decreases the value of $H$.
  So, given any $\x \subseteq E$ that is not a matching, we can remove problematic edges one at a time until we have a matching $\y$, and this sequence is monotone decreasing in $H$.
\end{proof}

\subsection{From non-maximal matching to maximal matching}
The next two results, \cref{lem:addOneEdge} and \cref{prop:MatchingVsMaxlMatching}, show that if you have a non-maximal matching, any sequence of adding edges while maintaining the matching property decreases the $H$ value.

\begin{lem}
    \label{lem:addOneEdge}
If $\y \subseteq E$ is a matching, and $\tilde \y$ differs from $\y$ by only adding an edge that is adjacent to no other edges of $\y$, then $H(\tilde \y) \leq H(\y)$.
\end{lem}

\begin{proof}
If $\y$ is a matching which is not maximal, then assume $(u,v) \in E$ exists with $u$ and $v$ adjacent to no edges in $\y$.
As this operation constitutes adding an edge, $F_c(\tilde \y) = F_c(\y) + \omega(u,v)$ by \cref{lem:CalculationRules}(b).

First, assume one of the vertices represents the diagonal; WLOG $v = \Delta_x$.
Then by \cref{lem:CalculationRules} rows (c) and (d),
\begin{equation*}
    F(\tilde \y) =
    \big( F_c(\y) + \omega(u,v) \big)
    + \big(F_U(\y) - B \big)
    +\big( F_V(\y)\big) < F(\y)
\end{equation*}
because $B > \omega(u,v)$.
If neither side represents the diagonal, then
\begin{equation*}
    F(\tilde \y) =
    \big( F_c(\y) + \omega(u,v) \big)
    + \big(F_U(\y) - B \big)
    +\big( F_V(\y) -B \big) < F(\y).
\end{equation*}
\end{proof}

%

\begin{prop}
\label{prop:MatchingVsMaxlMatching}
If $\y\subseteq E$ is a matching which is not maximal, then for some maximal matching $\z \supset \y$, $H(\y) > H(\z)$.
\end{prop}

\begin{proof}
  If $\y$ is a non-maximal matching, then we can add one edge to $\y$ for which each endpoint is not adjacent to any edges in $\y$.
  By \cref{lem:addOneEdge}, adding this edge to $\y$ decreases the value of $H$.
  So, given any non-maximal $\y$, we can add edges to $\y$ one at a time without losing the matching property until we arrive at a maximal matching $\z$.
  This is  a monotone decreasing sequence in $H$.
\end{proof}

\subsection{Proof of \cref{thm:MaxMatchingMinimizesH}}
With these tools in hand, we can turn to the proof of the main theorem.
\begin{proof}[Proof of \cref{thm:MaxMatchingMinimizesH}]
We prove the theorem by showing that any $\x \in Z$ (equivalently $\x \subseteq E$), which does not represent a MCMM has $H(\x) > H(\z)$ for any MCMM $\z \in Z$, equivalently $\z \subseteq E$.

Recall that if $\x \subseteq E$ represents a maximal matching, then $H(\x) = F_c(\x) = C(\x)$.
Of course, this means that if $\x$ were restricted to maximal matchings, $H(\x)$ would be minimized exactly when $\x$ represents a minimum cost matching.
However, the set $Z$ is larger than just maximal matchings.
So, we need to deal both with the case that $\x$ is not a matching, and that it is a matching, but is not maximal.

If $\x$ represents a matching that is not maximal, then by \cref{prop:MatchingVsMaxlMatching}, there is a maximal matching $\z$ for which $H(\x) > H(\z)$.
So, for any MCMM $\tilde \z$, $H(\x) > H(\z) \geq H(\tilde \z)$.
On the other hand, if $\x$ if not a matching, by combining \cref{prop:NonMatchingVsMatching} and \cref{prop:MatchingVsMaxlMatching}, there is a maximal matching $\z$ with $H(\x) > H(\y)$, and again we have $H(\x) > H(\z) \geq H(\tilde \z)$ for any MCMM $\tilde \z$.
%
%
\end{proof}

%% file: experiments.tex
\begin{figure}[t!b]
  \centering
  \includegraphics[width=0.95\textwidth]{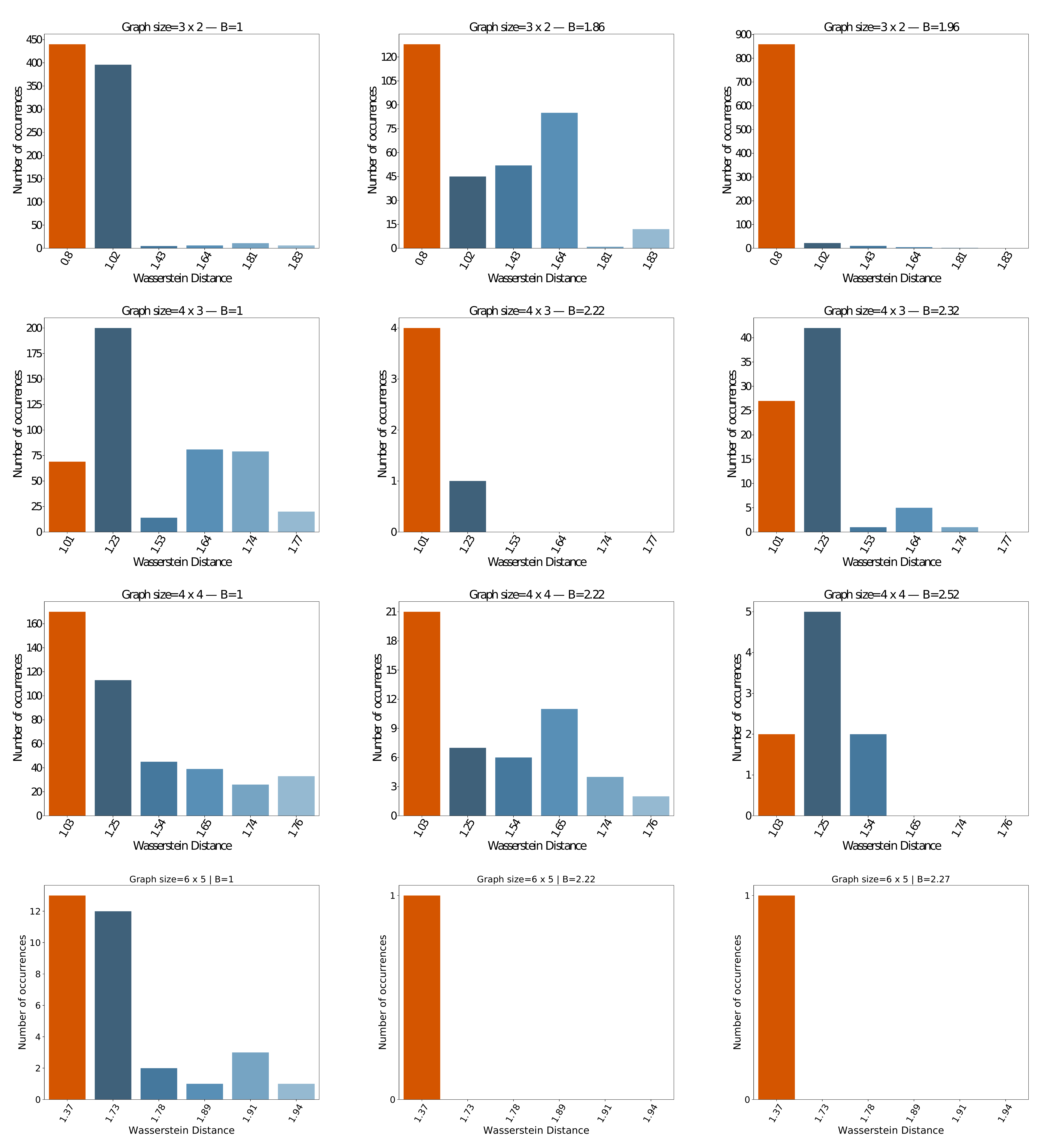}
  \caption{Squared Wasserstein distances ($x$-axis) and the number of times a sample was seen at each energy. Columns show results for different parameters, rows delineate problem size. See text for a discussion of the parameter range. Different problem sizes have different parameter ranges based on the optimal $B$. Correct values for each problem are given in \cref{tab:CorrectAnswers} and are denoted by orange bars in each histogram.}
  \label{fig:wass-dists-all}
\end{figure}

\section{Experiments}
\label{sec:experiments}

In this section we present results from computations performed on the D-Wave 2000Q~\cite{DWaveUserManual} quantum computer.
We consider squared Wasserstein distances computed from pairs of diagrams such as those in \cref{fig:dgm-example}.
We expect the QC to find the ground state corresponding to the MCMM, with higher energy states identifying the near-optimal matchings that still pair high-persistence points between diagrams.

We represent the graph $\tilde G$ as the QUBO $H$ in \cref{eq:H}.
\cref{thm:MaxMatchingMinimizesH} shows that to guarantee a MCMM in the context of the Wasserstein distance between two persistence diagrams, one must choose $B > B^* = \max_{(u,v) \in E(\widetilde{G})} \omega(u,v)$ (\cref{eq:DefB}), where it is understood that this value will change depending on the diagrams being compared.
Our experiments explore the dependence of the problem on the parameters of the Hamiltonian $B$ in \cref{eq:H} by testing output for $B = 1$, $B=B^*$ and $B=B^* + \e$ and  exploring the effect of $B$ on the low-energy solutions of $H$ returned by the QC.
We also analyze the problem over various diagram sizes.
Diagrams were obtained by taking subsets from diagrams similar to those of \cref{fig:dgm-example}, namely a diagram from an annulus and a diagram from a torus.
\begin{table}
    \centering
    \begin{tabular}{c|c}
        \textbf{Problem size } & \textbf{Wasserstein distance}\\
        $n \times m$& \textbf{$d_{2}^2(X,Y)$}\\\hline
        $3 \times 2$ & 0.802\\
        $4 \times 3$ & 1.012\\
        $4 \times 4$ & 1.03\\
        $6 \times 5$ & 1.37\\
    \end{tabular}
    \caption{Correct values of the squared Wasserstein distance for rows of \cref{fig:wass-dists-all} computed with Hera \cite{Kerber2016}.}
    \label{tab:CorrectAnswers}
\end{table}
Representative results are shown in \cref{fig:wass-dists-all}; rows correspond to diagram sizes and columns to values of $B$.
All computations are done with $p=q=2$.
The problem size is measured by the cardinality of the two diagrams used to construct $\widetilde{G}$.
The correct value of the squared Wasserstein distance for each row of \cref{fig:wass-dists-all} is given in \cref{tab:CorrectAnswers}. The values in the table match those in \cref{fig:wass-dists-all} where noted by the orange bars.
The middle column contains results from QUBOs where $B$ is set to $B^*$.
Note that only the rightmost column, therefore, fits with the assumptions of \cref{thm:MaxMatchingMinimizesH}.

%% file: conclusion.tex

\section{Conclusions}
\label{sec:conclusion}

The formulation of the Hamiltonian $H$ in \cref{eq:H} involves critical parameter choices. In particular, $B$ must be chosen with care.
It serves to balance the relative importance of the cost function $F_c$, with the constraint terms $F_U$ and $F_V$ in \cref{eq:H}.
For instance, setting $B$ too low will cause the quantum computer to violate the constraints.
In the extreme case of $B=0$, $H= F_c$ and there are no constraints on edge choices.
Hence, there is no energy penalty for setting all $x_{u,v} = 0$, resulting in $F_c = 0$. Thus, in this case, the quantum computer will favor non-maximal solutions.
Alternatively, if we set $B \gg 1$, the relative importance of the $F_c$ will be small.
While this results in samples that do not violate the constraints, the quantum computer will also fail to minimize the cost of matchings.

In our experiments, we observe that the quantum computer finds Wasserstein distances corresponding to the low-energy states correctly for small problems.
The quantum processing unit minimizes the Hamiltonians and finds the Wasserstein distance in many cases.
For graphs of size $3 \times 2$, shown in the top row in \cref{fig:wass-dists-all}, the quantum computer returns many samples at or near the MCMM edge configuration.
We can increase $B$ significantly past $B^*$ and still obtain correct samples from the quantum computer.

As the problem size increases, the ability of the quantum computer to consistently discover the low energy state across many parameters decreases.
In the second row, graphs of size $4 \times 3$, increasing $B$ quickly destabilizes the distribution of solutions.
This is possibly due to increasing necessity for long chains of physical qubits to encode high-degree nodes in the problem graph. Long chains have a higher chance of breaking, decreasing the possibility that the system will return numerous solutions at the ground state.
This corresponds to the number of {\em logical} qubits needed to formulate the problem. The node degree of $\widetilde{G}$ is such that the adjacency structure for even small problems is considered dense. Thus, computation of the Wasserstein distance becomes progressively more difficult to embed on the quantum processing unit, due to the need for a large number of physical qubits corresponding to each logical qubit.
We plan to investigate this issue in future research.